\documentclass[12pt]{article}
\usepackage{a4wide,hyperref}
\hypersetup{colorlinks=true, linkcolor=red,citecolor=blue,urlcolor=blue}
\usepackage{times}
\usepackage{latexsym}
\usepackage{amsmath}
\usepackage{amssymb}
\usepackage{ifthen}
\usepackage{mathrsfs}
\usepackage{graphics}
\usepackage{color}
\usepackage{tikz}
\usepackage{stmaryrd}
\usepackage{amsthm}

\newcommand{\nc}{\newcommand}
\nc{\rnc}{\renewcommand} \nc{\nev}{\newenvironment}
\rnc{\subsection}{\secdef\ssa\ssb}
\nc{\ssa}[2][default]{\par\vspace{1ex}\refstepcounter{subsection}\noindent\textbf{\thesubsection.
#1. }} \nc{\ssb}[1]{\par\vspace{2ex}\noindent\textbf{#1. }}

\rnc{\subsubsection}{\secdef\sssa\sssb}
\nc{\sssa}[2][default]{\par\vspace{1ex}\refstepcounter{subsubsection}\noindent\textit{\thesubsubsection.
#1. }} \nc{\sssb}[1]{\par\vspace{1ex}\noindent\textit{#1. }}

\makeatletter
\rnc{\@seccntformat}[1]{{\normalfont\bfseries{\csname
the#1\endcsname}\hspace{1pt}.\hspace{0.4em}}}
\rnc{\section}{\@startsection
        {section}%
        {1}%
        {0mm}%
        {-\baselineskip}%
        {0.5\baselineskip}%
        {\normalfont\normalsize\bfseries\centering}%
}
\renewcommand{\@makecaption}[2]{\begin{center}#1. #2\end{center}}
\makeatother

\newtheorem{theo}{Theorem}[section]
\newtheorem{lem}[theo]{Lemma}
\newtheorem{cor}[theo]{Corollary}
\newtheorem{prop}[theo]{Proposition}

\newtheorem{conj}[theo]{Conjecture}

\theoremstyle{definition}
\newtheorem{defn}[theo]{Definition}
\newtheorem{rem}[theo]{Remark}
\newtheorem{exa}[theo]{Example}

\rnc{\proof}[1][{}]{\smallskip\noindent\textit{Proof #1: }}
\nc{\proofend}{\hfill$\Box$\vspace{\topsep}\par}

\rnc{\labelenumi}{(\arabic{enumi})} \rnc{\labelitemi}{\text{--}}
\rnc{\phi}{\varphi} \rnc{\epsilon}{\varepsilon}
\nc{\bigmid}{\;\big|\;} \nc{\Bigmid}{\;\Big|\;}
\rnc{\max}{\textup{max}} \rnc{\min}{\textup{min}}
\nc{\tw}{\textup{tw}}

\newlength{\probwidth}
\nc{\prob}[3][9]{
\begin{center}
  \normalfont\fbox{
   \begin{tabular}[t]{
     rp{#1cm}}\textit{Instance:}&#2. \\
     \textit{Problem:}&#3
   \end{tabular}}
\end{center}}

\nc{\pprob}[4][9]{
\begin{center}
   \normalfont\fbox{
    \begin{tabular}[t]{
     rp{#1cm}}\textit{Instance:}&#2. \\
     \textit{Parameter:}&#3. \\
     \textit{Problem:}&#4
   \end{tabular}}
\end{center}}

\nc{\nprob}[4][11]{
\begin{center}
  \normalfont\fbox{

\addtolength{\probwidth}{#1cm}\parbox{\probwidth}{\textsc{#2}\\\hspace*{1.5em}
     \begin{tabular}[t]{
      rp{#1cm}}\textit{Instance:}&#3. \\
      \textit{Problem:}&#4
     \end{tabular}}}
\end{center}}

\nc{\npprob}[5][11]{
\begin{center}
  \normalfont\fbox{

\addtolength{\probwidth}{#1cm}\parbox{\probwidth}{\textsc{#2}\\\hspace*{1.5em}
    \begin{tabular}[t]{
     rp{#1cm}}\textit{Instance:}&#3. \\
     \textit{Parameter:}&#4. \\
     \textit{Problem:}&#5
    \end{tabular}}}
\end{center}}

\nc{\nppxrob}[5][9]{ \normalfont\fbox{

\addtolength{\probwidth}{#1cm}\parbox{\probwidth}{\textsc{#2}\\\hspace*{1.5em}
   \begin{tabular}[t]{
    rp{#1cm}}\textit{Instance:}&#3. \\
    \textit{Parameter:}&#4. \\
    \textit{Problem:}&#5
   \end{tabular}}}}

\nc{\nppprob}[5][4]{
\begin{center}
  \normalfont\fbox{

\addtolength{\probwidth}{#1cm}\parbox{\probwidth}{\textsc{#2}\\\hspace*{1.5em}
    \begin{tabular}[t]{
     rp{#1cm}}\textit{Instance:}&#3. \\
     \textit{Parameter:}&#4. \\
     \textit{Problem:}&#5
    \end{tabular}}}
\end{center}}

\nc{\noptprob}[6][9]{
\begin{center}
  \normalfont\fbox{

\addtolength{\probwidth}{#1cm}\parbox{\probwidth}{\textsc{#2}\\\hspace*{1.5em}
    \begin{tabular}[t]{
     rp{#1cm}}\textit{Instance:}&#3. \\
     \textit{Solution:}&#4. \\
     \textit{Cost:}&#5. \\
     \textit{Goal:}&#6.
    \end{tabular}}}
\end{center}}

\nc{\FOR}{\textbf{for}}
\nc{\FORALL}{\textbf{for all}}
\nc{\TO}{\textbf{to}}
\nc{\DO}{\textbf{do}}
\nc{\OD}{\textbf{od}}
\nc{\IF}{\textbf{if}}
\nc{\FI}{\textbf{fi}}
\nc{\THEN}{\textbf{then}}
\nc{\ELSE}{\textbf{else}}
\nc{\WHILE}{\textbf{while}}
\nc{\REPEAT}{\textbf{repeat}}
\nc{\UNTIL}{\textbf{until}}
\nc{\OR}{\textbf{or}}
\nc{\AND}{\textbf{and}}
\nc{\PRINT}{\textbf{print}}

\nc{\im}[1]{\item\hspace{#1cm}}
\nev{algorithm}{\begin{enumerate}\rnc{\labelenumi}{\textit{\small \arabic{enumi}.}}\rnc{\itemsep}{0ex}}{\end{enumerate}}

\nc{\fpcl}[1]{\left[#1\right]_{\text{\upshape fp}}}
\nc{\pr}{\le^{\text{\normalfont fp}}_m} \nc{\FPT}{\mathsf{FPT}}
\nc{\EPT}{\textup{EPT}} \nc{\SUBEPT}{\textup{SUBEPT}}

\nc{\fpt}{\textup{fpt}} \nc{\fptT}{\textup{fpt-T}}
\nc{\W}[1]{\text{$\textup{W}[#1]$}}
\nc{\M}[1]{\text{$\textup{M}[#1]$}}
\nc{\MS}[2]{\text{$\textup{M}^{#1}[#2]$}}
\nc{\MINI}[1]{\mbox{\small \rm MINI[$#1$]}}
\nc{\WP}{\textup{W[P]}} \nc{\AWP}{\textup{AW[P]}}

\rnc{\S}[1]{\text{$\textup{S}[#1]$}} \nc{\SP}{\textup{S[P]}}
\nc{\MP}{\textup{M[P]}}

\nc{\PTIME}{\mathsf{PTIME}}
\nc{\PH}{\mathsf{PH}}
\rnc{\P}{\mathsf{P}}
\nc{\APTIME}{\mathsf{APTIME}}
\nc{\PSPACE}{\mathsf{PSPACE}}
\nc{\NP}{\mathsf{NP}}

\nc{\DTIME}{\mathsf{DTIME}}

\nc{\se}{\subseteq} \nc{\re}{\rightarrow}

\nc{\LOEFF}[1]{{o}^{\rm eff}(#1)}

\nc{\PNPTC}{\mbox{$\textup{P}[{\textsc{tc}}]\ne
\textup{NP}[{\textsc{tc}}]$}}
\nc{\str}[1]{\ensuremath{\mathcal #1}}
\nc{\cls}[1]{\ensuremath{\mathbf #1}}

\nc{\algo}[1]{{\mathbb #1}}

\nc{\ceil}[1]{\left\lceil#1\right\rceil}
\nc{\floor}[1]{\left\lfloor#1\right\rfloor}

\nc{\bende}{\eqno$\Box$} \nc{\benda}{\tag*{$\Box$}}

\nc{\pa}{\kappa}

\nc{\co}{\mathsf{co}\text{-}}

\rnc{\L}{\mathsf{L}}

\nc{\NL}{\mathsf{NL}}

\rnc{\angle}[1]{\left\langle #1\right\rangle}
\nc{\dotcup}{\;\dot\cup\;}
\nc{\para}{\mathsf{para}\text{-}}

\newcommand{\MRDP}{\textup{MRDP}}

\newcommand{\FO}{\mathsf{FO}}

\newcommand{\ONE}{\mathit{ONE}}

\newcommand{\BIT}{\mathit{BIT}}

\newcommand{\inter}[1]{\ensuremath{\mathit #1}}

\newcommand{\C}{\mathsf C}

\newcommand{\dlogtime}{\textup{dlogtime}}

\newcommand{\E}{\mathsf{E}}
\newcommand{\NE}{\mathsf{NE}}
\newcommand{\LINH}{\mathsf{LINH}}

\newcommand{\AC}{\mathsf{AC}}
\newcommand{\paraAC}{\mathsf{para}\text{-}\mathsf{AC}^0}

\nc{\NLINSP}{\mathsf{NLINSPACE}}
\nc{\LINSP}{\mathsf{LINSPACE}}

\newcommand{\XAC}[1]{\mathsf{XAC}^0_{#1}}
\newcommand{\CAC}[1]{\mathsf{AC}^0_{#1}}

\newcommand{\pMC}{\ensuremath{p\text{-}\textsc{MC}}}

\newcommand{\phalt}{\ensuremath{p\text{-}\textsc{Halt}}}
\newcommand{\pdhalt}{\ensuremath{p\text{-}\textsc{DHalt}}}
\newcommand{\ptruth}{\ensuremath{p\text{-}\Delta_0\text{-}\textsc{Truth}}}
\newcommand{\pspec}{\ensuremath{p\text{-}\textsc{Spec}}}

\newcommand{\num}{\mathit{num}}
\newcommand{\un}{\mathit{un}}
\newcommand{\bin}{\mathit{bin}}

\newcommand{\arit}{\textup{ar}}
\newcommand{\rel}{\textup{r}}
\newcommand{\N}{\mathbb{N}}

\newcommand{\pMCFOA}{\pMC(L^\rel_\arit)}

\renewcommand{\le}{\leqslant}
\renewcommand{\ge}{\geqslant}

\title{A parameterized halting problem, $\Delta_0$ truth\\ and the \MRDP\ theorem%
\footnote{A partial conference version appeared as~\cite{chemueyok18}.}
}
\author{Yijia Chen\\
\small Department of Computer Science \\[-0.5ex]
\small Shanghai Jiao Tong University \\[-0.5ex]
\texttt{\small yijia.chen@cs.sjtu.edu.cn}
\and Moritz M\"uller \\
\small  Faculty of Computer Science and Mathematics  \\[-0.5ex]
\small University of Passau \\[-0.5ex]
\texttt{\small moritz.mueller@uni-passau.de}
\and Keita Yokoyama\\
\small Mathematical Institute \\[-0.5ex]
\small Tohoku University \\[-0.5ex]
\texttt{\small keita.yokoyama.c2@tohoku.ac.jp}}

\date{}

\begin{document}

\maketitle

\begin{abstract} We study the parameterized complexity of the problem to decide whether a given natural number $n$ satisfies a given $\Delta_0$-formula $\varphi(x)$; the parameter is the size of $\varphi$.
This parameterization focusses attention on instances where $n$ is large compared to the size of $\varphi$. We show unconditionally that this problem does not belong to the parameterized analogue of~$\AC^0$. From this we derive that certain  natural upper bounds on the complexity of our parameterized problem imply certain separations of classical complexity classes. This connection is obtained via an analysis of a parameterized halting problem. Some of these upper bounds follow assuming that $I\Delta_0$ proves the MRDP theorem in a certain weak sense.  
\end{abstract}

\section{Introduction}

\subsection{The parameterized halting problem}
The complexity of the following parameterized halting problem is still wide
open.
\npprob{\phalt}{$n\in \mathbb N$ in unary and a nondeterministic Turing
machine~$\mathbb M$} {$|\mathbb M|$, the size of $\mathbb M$}{Does $\mathbb
M$ accept the empty input in at most $n$ steps?}
The importance of \phalt\ is derived from its close connections to central
problems in proof complexity and descriptive complexity
theory~\cite{cheflu12,nasremvia05}. Among others, there is a logic for
$\PTIME$ if $\phalt$ can be decided by an algorithm in time $n^{f(|\mathbb
M|)}$ for some function $f: \mathbb N\to \mathbb N$. Sofar, however, such
algorithms have been ruled out only under a certain very strong
\emph{non-standard} complexity-theoretic hypothesis and only for computable
$f$~\cite{cheflu10,cheflu12}.

Thus, lower bounds on  $\phalt$ are poorly understood and of fundamental
interest. A seemingly modest and natural starting point is
\begin{conj}\label{conj:phaltAC}
$\phalt\notin \paraAC$.
\end{conj}
Here, $\paraAC$ is the analogue of (uniform) $\AC^0$ in the parameterized
world. One easily sees that $\phalt$ is in a \emph{nonuniform} version of
$\paraAC$: for fixed $k\in\N$, let $\mathbb M_{k,0}, \ldots, \mathbb
M_{k,\ell_k-1}$ list all nondeterministic Turing machines of size $k$ and let
$n_{k,i}$ be the minimal $n$ such that $\mathbb M_{k,i}$ accepts the empty
input in $n$ steps; if there is no such $n$, let $n_{k,i}:=\infty$. Then, on
instances $(1^n,\mathbb M)$ with parameter $|\mathbb M|=k$, $\phalt$ is
decided by the following family of simple Boolean functions:
\begin{align*}
F_{n,k}(& x_0\ldots x_{n-1}, y_0\ldots y_{k-1})
 = \bigvee_{\substack{\text{$i<\ell_k$ such} \\
     \text{that $n_{k,i}\le n$}}}
\big(x_0\ldots x_{n-1}= 1^n \wedge y_0\ldots y_{k-1}= \mathbb M_{k,i}\big).
\end{align*}
Observe that $F_{n,k}$ can be understood as a circuit of depth $2$ and size
$O(k\cdot \ell_k\cdot n)$.

Conjecture~\ref{conj:phaltAC} is highly plausible and might appear to be
within reach because $\AC^0$ is well-understood and,
in particular, \cite{cflowerac} establishes (unconditional) $\paraAC$ lower
bounds for many well-studied parameterized problems. It deserves some genuine
interest because its failure implies that $\AC^0$, or equivalently,
\emph{$(+, \times)$-invariant $\FO$} is captured by some logic. However, we
failed to prove the conjecture after years of attempts and only now
understand why: it implies that nondeterministic exponential time $\NE$ is
distinct from the linear time hierarchy $\LINH$. This connection can be
further tightened by considering the following variant of $\phalt$:
\npprob{$\phalt_=$}{$n\in \mathbb N$ in unary and a nondeterministic Turing
machine $\mathbb M$}{$|\mathbb M|$}{Does $\mathbb M$  accept the empty input
in \emph{exactly} $n$ steps.}

While the classical problems underlying $\phalt_=$ and $\phalt$ are easily
seen to be equivalent, we shall see that their parameterized versions behave
quite differently. In fact, $\phalt_=$ appears to be
harder than $\phalt$, e.g., a simple Boolean function family like $F_{n,k}$
for $\phalt$ is not known to exist for $\phalt_=$. We refer to
Section~\ref{sec:concl} for some discussion.

We show:

\begin{theo}\label{thm:phaltequal}\
\begin{enumerate}\itemsep=0pt
\item[(i)] $\phalt_=\in \paraAC$ if and only if $\NE\subseteq \LINH$.

\item[(ii)] $\phalt_=\in \paraAC$ implies $\phalt\in \paraAC$.
\end{enumerate}
\end{theo}

\subsection{The \MRDP\ theorem}
Thus, to settle Conjecture \ref{conj:phaltAC} one might try to first separate
$\NE$ from $\LINH$. We tie this question to the provability of the
Matiyasevich-Robinson-Davis-Putnam (\MRDP) theorem~\cite{mrdp} in bounded
arithmetic. This theorem states that $\Sigma_1$-definable sets are
\emph{Diophantine} and it is a long standing open problem whether it is
provable in $I\Delta_0$, i.e., Peano arithmetic with induction restricted to
$\Delta_0$-formulas.

Wilkie observed~\cite{wil80} that a positive answer would imply the collapse
of $\LINH$ to $\mathsf{NLIN}$ (nondeterministic linear time), and therefore
$\NP= \co\NP$ and $\NE\not\subseteq \LINH$. We derive the latter consequence
from an apparently much weaker provability assumption:

\begin{theo}\label{thm:mrdp}
If $I\Delta_0$ proves \MRDP\ for small numbers, then $\NE\not\subseteq
\LINH$.
\end{theo}

Roughly, that \emph{$I\Delta_0$ proves \MRDP\ for small numbers} means that
the equivalence of any $\Delta_0$-formula $\varphi(\bar x)$ to some
Diophantine formula is proved in $I\Delta_0$ for all $\bar x$ of logarithmic
order. Model-theoretically, the equivalence holds in any $I\Delta_0$-model
for all $\bar x$ from the initial segment of numbers $x$ such that $2^x$
exists, while proof-theoretically, we allow an
$I\Delta_0$-proof to use exponentiation, but only once.

Such limited use of exponentiation has been studied in bounded
arithmetic~\cite{kra}. An unlimited use of exponentiation is sufficient:
Gaifman and Dimitracopoulos~\cite{gaifdim82} showed that $I\Delta_0+\forall
x\exists y\ 2^x{=}y$ does prove \MRDP. Kaye~\cite{kaye} proved \MRDP\ using
only induction for bounded existential formulas plus an axiom stating the
totality of a suitable function of exponential growth. It is asked
in~\cite[p.188]{gaifdim82} whether $I\Delta_0$ plus the totality of $x^{\log
x}$, or of $x^{\log\log x}$ etc. proves \MRDP, and it is easy to see that if
the answer to the first (or any) of these questions is positive, then
$I\Delta_0$ proves \MRDP\ for small numbers. Intuitively, this provability is
much weaker than $I\Delta_0$-provability.

\subsection{$\Delta_0$ truth}
Theorem~\ref{thm:mrdp} follows from our main result concerning the complexity
to decide the truth of $\Delta_0$-sentences:
\npprob[10.5]{\ptruth}{$n\in \mathbb N$ in unary and a $\Delta_0$-formula
$\varphi(x)$}{$|\varphi|$, the size of $\varphi$}{$\N\models\varphi(n)$\ ?}
Intuitively, taking $|\varphi|$ as the parameter means shifting attention to
inputs where $n$ is much larger than $|\varphi|$. This is a natural focus.
Classical work of Paris and Dimitracopolous~\cite{paris} took $n$ to be
nonstandard and related the complexity of truth definitions for
$\Delta_0$-formulas to the complexity-theoretic hypotheses that $\LINH$ or
$\PH$ does not collapse.

Wilkie proved a weak version of the former hypothesis by showing that
$\ptruth$ restricted to quantifier-free formula inputs can be decided in
space $f(k)+O(\log n)$ where $k:= |\varphi|$ is the parameter and $f:\N\to
\N$ a computable function~\cite[Proof of Lemma 3.1]{wil80}.
The straightforward algorithm decides $\ptruth$ in space $f(k)\cdot \log n$.
Can $\ptruth$  be decided in space $f(k)+ O(\log n)$? Maybe with
nondeterminism? Can it be decided in time $f(k)\cdot n^{O(1)}$? Maybe with
nondeterminism, i.e., is $\ptruth\in
\para\NP$?

At present all these questions are wide open. Our main result
(Theorem~\ref{thm:truthupper}) shows that such \emph{upper bounds} on the
parameterized complexity of $\ptruth$ imply \emph{lower bounds} in classical
complexity theory. Notably,

\begin{theo}\label{thm:truthNP}
If $\ptruth\in \para\NP$, then $\NE\not\subseteq \LINH$.
\end{theo}

Theorem~\ref{thm:truthNP} follows from our analysis of $\phalt_=$ and the
following unconditional lower bound:
\begin{theo}\label{thm:truthAC}
$\ptruth\not\in\paraAC$.%
\end{theo}
The proof is based on diagonalization or, more specifically, the
undefinability of truth. Furthermore, it relies on the classical
result~\cite{BIS} of descriptive complexity theory that, roughly speaking,
equates $\AC^0$ and first-order logic with built-in arithmetic.

\subsection{$\AC^0$-bi-immunity}
Could Conjecture~\ref{conj:phaltAC} be false? We give further evidence for
its truth by establishing a connection to the existence of $\AC^0$-bi-immune
sets in $\NP$. Recall, a problem $Q$ is \emph{$\AC^0$-bi-immune} if neither
$Q$ nor its complement contain an infinite subset in $\AC^0$.

\begin{theo}\label{thm:immu}
If $\NP$ contains an $\AC^0$-bi-immune problem, then $\phalt\not\in\paraAC$.
\end{theo}
It is a standard hypothesis that $\NP$ contains even $\P$-bi-immune problems
and this follows from the measure hypothesis~\cite{mayor}. Whether $\NP$
contains at least $\AC^0$-bi-immune problems has been asked once it was
realized~\cite{ghs,abhh} that deterministic time hierarchy theorems hold with
bi-immunity (or, equivalently~\cite{schoening}, almost everywhere) while this
is open for nondeterministic time~\cite{abhh,forsan}.
While Zimand~\cite{zimand} obtained some partial positive answers, Allender
and Gore~\cite{allgor90} showed that this has different answers relative to
different oracles.\footnote{\cite{allgor90} studies $\AC^0$-immunity but
their oracle constructions can be adapted to $\AC^0$-bi-immunity.} This
indicates that also refuting  Conjecture~\ref{conj:phaltAC} might be
non-trivial.

\subsection{Outline}
Much of the technical work consists in connecting the dots between results of
various subareas of logic and complexity, namely classical, parameterized and
descriptive complexity theory and formal arithmetic. Section~\ref{sec:prelim}
reviews the results we need and fixes our notation. The technicalities are
somewhat subtle, in particular, the move from $\phalt$ to $\phalt_=$ is
crucial. Section~\ref{sec:halt} proves Theorem~\ref{thm:phaltequal} and
various variants of it via an analysis of $\phalt_=$: it is in a strong sense
the hardest among what we call \emph{almost tally} problems in $\para\NP$.
Such problems have instances consisting of a natural number in unary notation
plus a short binary string where short means having a length effectively
bounded in terms of the parameter. Section~\ref{sec:truth} proves
Theorem~\ref{thm:truthAC}. This together with  the results in
Section~\ref{sec:halt} implies Theorem~\ref{thm:truthNP} and various
variants. Section~\ref{sec:mrdp} derives (a strengthening of)
Theorem~\ref{thm:mrdp} from Theorem~\ref{thm:truthNP}. Section~\ref{sec:immu}
proves Theorem~\ref{thm:immu}. The final section discusses the role of
uniformity, and exhibits the different behaviours of our parameterized
problems $\phalt$, $\phalt_=$ and $\ptruth$.


\section{Preliminaries}\label{sec:prelim}

Standard monographs are~\cite{papad,arora} for classical complexity theory,
\cite{df,fg,df2} for parameterized complexity theory, \cite{hp,kayebook} for
formal arithmetic, and~\cite{immer,ebbflu} for descriptive complexity theory.

\subsection{Classical complexity}
A \emph{(classical) problem} is a subset of $\{0,1\}^*$, the set of finite
binary strings. The length of a binary string $x\in\{0,1\}^*$ is denoted
$|x|$. For $n\in \mathbb N$ we let $1^n$ denote the binary string consisting
of $n$ many~$1$'s. We use multitape Turing machines with alphabet $\{0,1\}$
as our basic model of computation. When considering \emph{dlogtime} Turing
machines, i.e.\ deterministic machines running in time $O(\log n)$, it is
understood that they access their input via an address tape
(see~e.g.~\cite{BIS}). As usual, $\P$ and $\NP$ denote deterministic and
nondeterministic polynomial time~$n^{O(1)}$, and $\E$ and $\NE$ denote
deterministic and nondeterministic exponential time with linear exponent
$2^{O(n)}$. The \emph{linear time hierarchy} $\LINH$ is the set of problems
acceptable by alternating Turing machines in linear time $O(n)$ with $O(1)$
alternations. $\LINSP$ and $\NLINSP$ denote deterministic and
nondeterministic linear space $O(n)$. Clearly,
\[
\LINH\subseteq \LINSP\subseteq \NLINSP\subseteq \E\subseteq \NE.
\]

Following~\cite{BIS} we define (dlogtime uniform) $\AC^0$ as the set of
problems decided by \emph{$\AC^0$-circuit families} $\big(\C_n\big)_{n\in
\mathbb N}$:
\begin{enumerate}\itemsep=0pt
\item[--] $\C_n$ is a circuit (with $\wedge, \vee, \neg$ gates and
    unbounded fan-in) with $n$ variables, size $\le n^c$ and depth $\le d$,
    where $c,d\in \mathbb N$ are two constants independent of $n$;

\item[--] there is a \dlogtime\ Turing machine which given $\angle{1^n, i,
    b}$ where $n,i\in \mathbb N$ and $b\in \{0,1\}$ decides whether the
    $i$-th bit of the binary encoding of $\C_n$ is $b$.
\end{enumerate}
Here, for binary strings $x= x_0\cdots x_{|x|-1}$ and $y= y_0\cdots
y_{|y|-1}$ we use the standard pairing
\begin{equation}\label{eq:pairing}
\angle{x,y}:=
x_0x_0\cdots x_{|x|-1}x_{|x|-1}01y_0y_0\cdots y_{|y|-1}y_{|y|-1},
\end{equation}
and similarly for more arguments. The above definition is somewhat sensitive
to the choice of the binary encoding of a circuit. An appropriate choice
would be to encode~$C_n$ by the list of strings in the \emph{direct
connection language} corresponding to $n$; we refer to~\cite{BIS} for
details.

For $n\in\N$ we let $\bin(n)\in\{0,1\}^*$ denote the binary expansion of $n$;
it has length $\ceil{\log(n+1)}$. For $x\in \{0,1\}^*$ let $\num(x)$ be the
natural number with binary expansion $1x$, i.e., $\bin(\num(x))= 1x$. For a
problem $Q$ let
\[
\un(Q):= \left\{1^{\num(x)}\mid x\in Q\right\}.
\]
The last statement of the following is~\cite[Proposition~5]{allgor90}, and
the first two are trivial:

\begin{prop}[\cite{allgor90}]\label{prop:ag}
Let $Q$ be a problem. Then:
\begin{enumerate}\itemsep=0pt
\item[(i)] $Q\in \NE$ if and only if $\un(Q)\in \NP$.

\item[(ii)] $Q\in \E$ if and only if $\un(Q)\in \P$.

\item[(iii)] $Q\in \LINH$ if and only if $\un(Q)\in \AC^0$.
\end{enumerate}
\end{prop}


\subsection{Parameterized complexity}
A \emph{parameterized problem} is a pair $(Q, \kappa)$ of an
\emph{underlying} classical problem $Q\subseteq \{0,1\}^*$ and a polynomial
time computable \emph{parameterization} $\kappa:\{0,1\}^*\to \mathbb N$
mapping an instance $x\in \{0,1\}^*$ to its \emph{parameter} $\kappa(x)\in
\mathbb N$.
E.g., $\phalt$ has underlying classical problem $\big\{\angle{1^n, \mathbb
M}\bigmid$ the nondeterministic Turing machine $\mathbb M$ accepts the empty
input in at most $n$ steps$\big\}$ and a parameterization $\kappa$ that maps
strings of the form $\angle{1^n, \mathbb M}$ to $|\mathbb M|$ and other
strings to, say,~0.

The para-operator~\cite{flugro03} turns a classical complexity class into a
parameterized one (the most important intractable parameterized classes are
not of this form, however). The class $\para\P= \FPT$ contains the
parameterized problems $(Q,\kappa)$ that are {\em fixed-parameter tractable},
i.e., decidable in deterministic time $f(\kappa(x))\cdot |x|^{O(1)}$ for some
computable $f:\N\to\N$. Similarly, $\para\NP$ denotes nondeterministic time
$f(\kappa(x))\cdot |x|^{O(1)}$ (for any computable $f$), $\para\L$ denotes
deterministic space $f(\kappa(x))+O(\log|x|)$, and $\para\NL$ denotes
nondeterministic such space. Clearly,
$$
\para\L\subseteq\para\NL\subseteq\FPT\subseteq\para\NP.
$$

The central parameterized
 class in this paper is $\paraAC$. It is characterized as follows:

\begin{prop}[\cite{cflowerac}]\label{prop:paraAC}
Let $(Q, \kappa)$ be a parameterized problem such that $Q$ is decidable and
$\kappa$ is computable by an $\AC^0$-circuit family. Then the following are
equivalent.
\begin{enumerate}\itemsep=0pt

\item[(i)] $(Q,\kappa)\in \paraAC$.

\item[(ii)] There is a family $(\C_{n,k})_{n,k\in \mathbb N}$ of circuits
    such that

    \begin{enumerate}\itemsep=0pt
    \item[--] there are a computable $f:\mathbb N\to \mathbb N$ and
        $c,d\in \mathbb N$ such that for all $n,k\in \mathbb N$ the
        circuit $\C_{n,k}$ has $n$ inputs, size at most $ f(k)\cdot n^c$,
        and depth at most $d$;

    \item[--] for all $x\in \{0,1\}^*$ we have
        \[
        x\in Q\iff \C_{|x|, \kappa(x)}(x)=1;
        \]

    \item[--] there are a computable $g:\mathbb N\to \mathbb N$ and a
        deterministic Turing machine which given as input $\angle{1^n,
        1^k, i, b}$ where $n,k,i\in \mathbb N$ and $b\in \{0,1\}$ decides
        in time $g(k)+ O(\log n)$ whether the $i$-th bit of the binary
        encoding of $\C_{n,k}$ is $b$.
    \end{enumerate}

\item[(iii)] There are a computable $h: \mathbb N\to \mathbb N$ and an
    $\AC^0$-circuit family $(\C_n)_{n\in \mathbb N}$  such that for all
    $x\in \{0,1\}^*$ with $|x|\ge h(\kappa(x))$ we have
    \[
    x\in Q\iff \C_{|x|}(x)=1.
    \]
\end{enumerate}
\end{prop}

According to the terminology of~\cite{flugro03}, (iii) states that
$(Q,\kappa)$ is \emph{eventually in $\AC^0$}.


\subsection{Formal arithmetic}
We let $L_\arit:=\{+,\times,0,1,<\}$ be the language of arithmetic with
binary function symbols $+,\times$, constants $0,1$ and a binary relation
symbol $<$. The {\em standard $L_\arit$-structure}, denoted $\N$, has
universe $\N$ and interprets the symbols in the obvious way. Every
$L_\arit$-term $p$ computes a polynomial with coefficients in $\N$ and of
total degree at most $|p|$. We do not distinguish terms $p$ or formulas
$\varphi$ from their binary encodings, so $|p|$ and $|\varphi|$ denote the
lengths of these encodings. Writing $\varphi(\bar x)$ for a formula $\varphi$
means that \emph{all} free variables of~$\varphi$ are among~$\bar x$. A {\em
sentence} is a formula without free variables.

A \emph{$\Delta_0$-formula} is an $L_\arit$-formula obtained from atomic
formulas, Boolean connectives, and bounded quantifiers $\exists x{<}p$,
$\forall x{<}p$ where $p$ is an $L_\arit$-term not involving $x$; e.g.,
$\exists x{<}p\; \varphi$ stands for $\exists x (x{<}p\wedge \varphi)$.
$\Sigma_1$- and $\Pi_1$-{\em formulas} are obtained from $\Delta_0$-formulas
by existential and  universal quantification, respectively.

\begin{theo}[\MRDP]
For every $\Delta_0$-formula $\varphi(\bar x)$ there are $L_\arit$-terms
$p(\bar x,\bar y),q(\bar x,\bar y)$ such that
\[
\N\models \forall\bar x\big(\varphi(\bar x)\leftrightarrow
 \exists\bar y\; p(\bar x,\bar y){=}q(\bar x,\bar y)\big).
\]
\end{theo}

G\"odel showed that computable functions are $\Sigma_1$-definable. The \MRDP\
theorem improves this to an existential definition:

\begin{cor}\label{cor:fxy}
For every computable $f:\mathbb N\to \mathbb N$ there is a quantifier-free
$L_\arit$-formula $\varphi_f(x, y, \bar z)$ such that for every $n,m\in
\mathbb N$
\begin{eqnarray*}
f(n)= m & \iff &
\N \models \exists \bar z\; \varphi_f(n,m,\bar z).
\end{eqnarray*}
\end{cor}

We are mainly concerned with finite arithmetical structures with universe
\[
[n]:=\{0,\ldots, n-1\}
\]
for some $n\in\N$ with $n\ge 2$, and
therefore consider the relational version
\[
L_\arit^\rel
\]
of $L_\arit$ where $+$, $\times$ are ternary relation symbols. The standard
$L_\arit^\rel$-structure with universe $\N$, also denoted $\N$, interprets
$+$, $\times$ by the graphs of addition and multiplication, respectively. For
$n\in \N$ with $n\ge 2$, the standard $L^\rel_\arit$-structure with universe
$[n]$, simply denoted $n$, is the substructure of $\N$ with universe $[n]$,
i.e., it interprets the symbols in $L^\rel_\arit$ by $+^n:=\{(k,\ell,m)
\in[n]^3\mid k+\ell=m\}$, $\times^n:= \{(k,\ell,m)\in [n]^3\mid k\cdot \ell=
m\}$, $0^n:=0$, $1^n:=1$ and $<^n:=\{(k,\ell)\in[n]^2\mid k<\ell\}$.

For every $L^\rel_\arit$-formula $\varphi(\bar x)$ and $\bar n, n\in \N$ with
$1, \bar n<n$ we have
\begin{eqnarray*}
\N\models \varphi^{<n}(\bar n)
 & \iff &
 n\models \varphi(\bar n),
\end{eqnarray*}
where $\varphi^{<y}$ is obtained from $\varphi$ by replacing all quantifiers
$\exists z$, $\forall z$ by $\exists z{<}y$, $\forall z{<}y$.

\begin{rem}\label{rem:fxy}
Corollary~\ref{cor:fxy} holds for a quantifier-free $L^\rel_\arit$-formula
$\varphi_f(\bar x, y, \bar z)$. Indeed, it is straightforward to express an
$L_\arit$-term equality by an existential $L^\rel_\arit$-formula.
\end{rem}

\subsection{Descriptive complexity}
A binary string $x= x_0\cdots x_{n-1}\in \{0,1\}^*$ of length $n\ge 2$ is
often identified with the \emph{string structure} $\str S(x)$ defined as the
$L^\rel_\arit\cup\{\ONE\}$-expansion of the standard $L^\rel_\arit$-structure
$n$ that interprets the unary relation symbol $\ONE$ by
\[
\ONE^x:=\{i\in[n]\mid x_i=1\},
\]
i.e., $\str S(x)= \big([n], +^n, \times^n, 0^n, 1^n, <^n, \ONE^x\big)$.
We shall work with the following descriptive characterization of (dlogtime
uniform) $\AC^0$:
\begin{theo}[\cite{BIS}]\label{thm:BIS}
A problem $Q$ is in $\AC^0$ if and only if there is an
$L^\rel_\arit\cup\{\ONE\}$-sentence~$\varphi$ such that for every $x\in
\{0,1\}^*$ with $|x|\ge 2$:
\begin{eqnarray*}
x\in Q & \iff & \str S(x)\models \varphi.
\end{eqnarray*}
\end{theo}

This result and Proposition~\ref{prop:paraAC}~(iii) imply what is to be our
working definition of $\paraAC$: the parameterized problems that are
\emph{eventually definable}.

\begin{cor}\label{cor:evtfo}
Let $(Q, \kappa)$ be a parameterized problem such that $Q$ is decidable and
$\kappa$ is computable by an $\AC^0$-circuit family.
Then $(Q, \kappa)$ is in $\paraAC$ if and only if $(Q,
\kappa)$ is \emph{eventually definable}: there are a computable $h: \N\to \N$
and an $L^\rel_\arit\cup \{\ONE\}$-sentence $\varphi$ such that for all $x\in
\{0,1\}^*$ with $|x|\ge h(\kappa(x))$:
\begin{eqnarray*}
x\in Q & \Longleftrightarrow & \str S(x)\models \varphi.
\end{eqnarray*}
\end{cor}

In descriptive complexity the role of reductions is played by
interpretations. Let $L,L'$ be languages consisting of relation symbols and
constants. Let $w\in \N$ with $w\ge 1$. An \emph{interpretation $\inter I$ of
$L'$ in $L$ (of width $w$)} is given by an $L$-formula
$\varphi_{\textup{uni}}(\bar x)$, an $L$-formula $\varphi_R(\bar x_0, \cdots,
\bar x_{r-1})$ for each $r$-ary relation symbol $R\in L'$, and an $L$-formula
$\varphi_c(\bar x)$ for every constant $c\in L'$; here, $\bar x$, $\bar x_i$
are $w$-tuples of variables. Given an $L$-structure $A$ define the
$L'$-structure $A^{\inter I}$ as follows. It has universe $A^{\inter I}:=
\big\{\bar a\in A^w\bigmid A\models \varphi_{\textup{uni}}(\bar a)\big\}$,
interprets an $r$-ary $R\in L'$ by $\big\{(\bar a_0,\ldots,\bar a_{r-1})\in
(A^{\inter I})^r \bigmid A\models \varphi_R(\bar a_0, \ldots, \bar a_{r-1})
\big\}$, and a constant $c\in L'$ by the unique $\bar a\in A^{\inter I}$
satisfying $\varphi_c(\bar x)$ in $A$. If this uniqueness is violated or if
the universe $A^{\inter I}$ is empty, then $A^{\inter I}$ is not defined. If
$B\cong A^I$ for some~$I$, we say $B$ is {\em interpretable} in $A$. The
following is standard.
\begin{lem}\label{lem:interpretation}
Let $\inter I$ an interpretation of $L'$ in $L$ of width $w$ and $I'$ an
interpretation of $L''$ in $L'$ of width $w'$. Further let $A$ be an
$L$-structure such that $A^I$ is defined.
\begin{enumerate}\itemsep=0pt
\item[(i)] For every $L'$-formula $\varphi(x,y,\ldots)$ there is an
    $L$-formula $\varphi^I(\bar x,\bar y,\ldots)$ where $\bar x, \bar y,
    \ldots$ are $w$-tuples of variables such that for all $\bar a, \bar b,
    \ldots \in A^I$:
    \begin{eqnarray*}
    A^I \models \varphi(\bar a, \bar b, \ldots)
     & \iff &
    A \models \varphi^I(\bar a, \bar b, \ldots).
    \end{eqnarray*}

\item[(ii)] There is an interpretation $I'\circ I$ of $L''$ in $L$ of width
    $w\cdot w'$ such that if $(A^I)^{I'}$ is defined, then so is
    $A^{I'\circ I}$ and
    \[
    A^{I'\circ I }\cong (A^I)^{I'}.
    \]
\end{enumerate}
\end{lem}

\medskip
The following is folklore, and a proof can be found
in~\cite[Appendix]{sch05}.
\begin{lem}\label{lem:karyinterpretation}
Let $d\in\N$.
\begin{enumerate}\itemsep=0pt
\item[(i)] For every $n\ge 2$ the standard $L^\rel_\arit$-structure $n^d$
    is interpretable in the standard $L^\rel_\arit$-structure~$n$. In fact,
    there is an interpretation $\inter I_d$ of width $d$ such that
    $n^d\cong n^I$ for every $n\ge 2$, and the isomorphism maps each
    $a<n^d$ to the length $d$ representation of $a$ in base~$n$.

\item[(ii)] There is an $L^\rel_\arit$-formula $\BIT(x, y)$ such that for
    every $n\ge 2$ and all $i, j\in [n]$:
    \begin{eqnarray*}
    n\models \BIT(i,j)
     & \iff &
    \text{the $j$-th bit of $\bin(i)$ is $1$}.
    \end{eqnarray*}
\end{enumerate}
\end{lem}

\section{$\phalt$ and $\NE$ versus $\LINH$}\label{sec:halt}

In this section we first introduce a workable notion of reduction that
preserves $\paraAC$, prove Theorem~\ref{thm:phaltequal}, and then consider
some generalizations and variants that will be instrumental later in
Section~\ref{sec:truth} for the proof of Theorem~\ref{thm:truthNP} and its
variants.

\subsection{Eventually definable reductions}
A \emph{parameterized} reduction from a parameterized problem $(Q,\kappa)$ to
another $(Q',\kappa')$ is a reduction $r:\{0,1\}^*\to\{0,1\}^*$ from $Q$ to
$Q'$ such that $\kappa'\circ r\le f\circ \kappa$ for some computable function
$f: \N\to \N$.

\begin{defn}\label{def:kappaevt}
Let $\kappa$ be a parameterization. A function $r:\{0,1\}^*\to \{0,1\}^*$ is
\emph{$\kappa$-eventually definable} if there are a computable $h:\N\to\N$
and an interpretation~$\inter I$ such that
\begin{eqnarray*}
\str S(x)^{\inter I} & \cong & \str S(r(x))
\end{eqnarray*}
for all $x\in\{0,1\}^*$ with $|x|\ge h(\kappa(x))$.
\end{defn}


\begin{exa}\label{exa:internum}
The function
\[
\angle{1^n, x} \mapsto
 1^{\displaystyle \num(\langle \bin(n),x\rangle)}
\]
where $n\in \N$, $x\in \{0,1\}^*$ is $\kappa$-eventually definable where
$\kappa$ maps $\angle{1^n, x}$ to $|x|$ (both functions map arguments that
are not of the required form to, say, $0$).
\end{exa}

\proof Note $\num(\angle{\bin(n), x})< 2^{|\angle{\bin(n), x}|+ 1}\le
2^{O(\log n+ |x|)}$. Choose a constant $d\in \N$ and a computable $h: \N\to
\N$ such that $\num(\angle{\bin(n), x})< n^d$ and $\num(x)< n$ whenever $n\ge
h(|x|)$. We describe an interpretation of $\str S(1^{\num(\angle{\bin(n),
x})})$ in $\str S(\angle{1^n, x})$ whenever $n\ge h(|x|)$. It will be clear
that the interpretation does not depend on $n$, $x$.

Let $(n,\num(x))$ be the expansion of the standard $L^\rel_\arit$-structure
$n$ that interprets a new constant by $\num(x)\in[n]$. This is interpretable
in $\str S(\angle{1^n, x})$ using $\BIT$. By
Lemma~\ref{lem:karyinterpretation}, also $(n^d, \num(x))$ is interpretable in
$\str S(\angle{1^n, x})$. But this structure defines ($n$ and)
$\num(\angle{\bin(n), x})\in [n^d]$ using $\BIT$. Thus, $\str
S(1^{\num(\angle{\bin(n), x})})$ is interpretable in $\str S(\angle{1^n, x})$
as claimed. \proofend

Recall, a function $r:\{0,1\}^*\to\{0,1\}^*$ is \emph{honest} if $|r(x)|\ge
|x|^{\Omega(1)}$.

\begin{lem}\label{lem:transitive}
Assume that $r,r': \{0,1\}^*\to \{0,1\}^*$ are $\kappa$- and
$\kappa'$-eventually definable, respectively, that $\kappa'\circ r\le f\circ
\kappa$ for some computable $f:\N\to \N$, and that $r$ is honest. Then
$r'\circ r$ is $\kappa$-eventually definable.
\end{lem}

\proof Choose $\inter I,h$ for $r$ and $\inter I',h'$ for $r'$ according to
Definition~\ref{def:kappaevt}. We can assume
that~$h'$ is nondecreasing. Choose $n_0, c\in\N$ such that $|r(x)|\ge
|x|^{1/c}$ for all $x\in\{0,1\}^*$ with
$|x|\ge n_0$. Define $g:\N\to\N$ by
\[
g(k):= \max\big\{h'(f(k))^c,h(k),n_0\big\}.
\]
We claim that $I'\circ I$ and $g$ witness that $r'\circ r$ is
$\kappa$-eventually definable. To verify this let $x\in\{0,1\}^*$ satisfy
$|x|\ge g(k)$ where $k:= \kappa(x)$. Then $|r(x)|\ge |x|^{1/c}\ge h'(f(k))\ge
h'(\kappa'(r(x)))$ using that~$h'$ is nondecreasing. Hence $\str
S(r(x))^{I'}\cong \str S(r'(r(x)))$. Then $\str S(x)^{I'\circ I}\cong \str
S(r'(r(x)))$ because $\str S(x)^I\cong\str S(r(x))$, because $|x|\ge h(k)$.
\proofend

\begin{defn}\label{defn:evtdef}
Let $(Q, \kappa)$ and $(Q', \kappa')$ be parameterized problems. An
\emph{eventually definable reduction from $(Q, \kappa)$ to $(Q', \kappa')$}
is a parameterized reduction from $(Q, \kappa)$ to $(Q', \kappa')$ that is
honest and $\kappa$-eventually definable.
\end{defn}

\begin{rem}\label{rem:evtdef}
A parameterized problem is in $\paraAC$ if and only if there is an eventually
definable reduction from $(Q,\kappa)$ to a trivial problem, say, $(Q_0,
\kappa_0)$ for $Q_0$ the set of strings starting with $0$ and $\kappa_0$ is
constantly $0$.
\end{rem}

It is straightforward to check that this reducibility is transitive and
preserves membership in~$\paraAC$:
\begin{lem}\label{lem:ACclosed}
Assume there is an eventually definable reduction from $(Q, \kappa)$ to $(Q',
\kappa')$.
\begin{enumerate}\itemsep=0pt
\item[(i)] If there is an eventually definable reduction from $(Q',
    \kappa')$ to $(Q'', \kappa'')$, then there is one from $(Q, \kappa)$ to
    $(Q'', \kappa'')$.

\item[(ii)]If $(Q', \kappa')\in \paraAC$,  then $(Q, \kappa)\in \paraAC$.
\end{enumerate}
\end{lem}

\proof
%
%
(i) follows from Lemma~\ref{lem:transitive}. (ii) follows from (i) and
Remark~\ref{rem:evtdef}. \proofend

\subsection{The complexity of $\phalt_=$}
It is known that the question whether $\phalt_=$ is fixed-parameter tractable
is closely related to the relationship of $\E$ and $\NE$:

\begin{theo}[\cite{aumdom08,cheflu09}]\label{thm:phaltENE}
$\phalt_=\in \FPT$ if and only if $\NE\subseteq \E$.
\end{theo}

Theorem~\ref{thm:phaltequal}~(i) is a $\paraAC$-analogue of this result.

\medskip\noindent \emph{Proof of Theorem~\ref{thm:phaltequal}:} (ii) follows easily
from the equivalence that a nondeterministic Turing machine $\mathbb M$
accepts the empty input in at most $n$ steps if and only if $\mathbb M$
accepts the empty input in exactly $n'$ steps for some $n'\le n$.

To prove (i), first assume $\NE\subseteq \LINH$ and let $Q$ be the classical
problem underlying $\phalt_=$ but with input $n$ encoded in binary:
\nprob{$Q$}{$n \in \mathbb N$ in \emph{binary} and a nondeterministic Turing
machine $\mathbb M$}{Does $\mathbb M$ accept the empty input in exactly $n$
steps?}
Clearly, $Q\in \NE$, so by assumption and Proposition~\ref{prop:ag}~(iii) we
have $\un(Q)\in \AC^0$. Recall
\begin{align*}
\un(Q)=
 \Big\{1^{\num(\angle{\bin(n), \mathbb M})} \Bigmid
 \begin{array}{l}
 \text{the nondeterministic Turing machine $\mathbb M$} \\
  \text{accepts the empty input
   in exactly $n$ steps}
  \end{array}
 \Big\}.
\end{align*}
By Example~\ref{exa:internum} the map $\angle{1^n, \mathbb M} \mapsto
1^{\num(\angle{\bin(n), \mathbb M})}$ is eventually definable with respect to
the parameterization of $\phalt_=$. It is a honest parameterized reduction to
$(\un(Q), \kappa)$ where~$\kappa$ maps $1^{\num(\angle{\bin(n), \mathbb M})}$
to $|\mathbb{M}|$. Since $(\un(Q), \kappa)\in \paraAC$,
Lemma~\ref{lem:ACclosed} implies $\phalt_=\in \paraAC$.

\medskip
Conversely, assume $\phalt_=\in \paraAC$. Let $Q\subseteq \{0,1\}^*$ be a
problem in $\NE$. To show that $Q\in \LINH$, it suffices to prove $\un(Q)\in
\AC^0$ again by Proposition~\ref{prop:ag}~(iii).

As $Q \in \NE$ there is a nondeterministic Turing machine $\mathbb M$ and a
constant $c \in \mathbb N$ such that $\mathbb M$ accepts $Q$ in time at most
$\num(x)^c-2|x|-2$. Consider the nondeterministic
Turing machine~$\mathbb M^*$\label{page:Mstar} that started with the empty
input runs as follows:
\begin{center}
\fbox{
\begin{minipage}[t]{15cm}
\begin{algorithm}
\im0 guess $y\in\{0,1\}^*$

\im0 simulate $\mathbb M$ on $y$

\im0 \IF\ $\mathbb M$ rejects, \THEN\ reject

\im0 make dummy steps such that so far the total
 running time  is $\num(y)^c$

\im0 accept.
\end{algorithm}
\end{minipage}
}
\end{center}
Line~1 takes exactly $2|y|+2$ many steps by moving the head forth and back on
some tape, 
so the dummy steps in line~4 are possible. Since $\num$ is injective, we have
\begin{eqnarray}\label{eq:mes}
x \in Q  & \iff & \text{$\mathbb M^*$ accepts the empty input tape
in exactly $ \num(x)^c+1$ many steps}.
\end{eqnarray}
Since $\mathbb M^*$ is a fixed machine,   $\phalt_=\in \paraAC$ implies that
the classical problem
\[
Q':=\Big\{ 1^n\mid \text{$\mathbb M^*$ accepts the empty input tape
in exactly $n+1$ many steps}\Big\}
\]
is in $\AC^0$. Choose a first-order sentence $\varphi$ for $Q'$ according to
Theorem~\ref{thm:BIS}. Lemma~\ref{lem:karyinterpretation} gives an
interpretation $\inter I$ such that $\str S(1^n)^{\inter I}\cong \str
S(1^{n^c})$ for all $n\ge 2$. Then $1^{n^c}\in Q'$ is equivalent to $\str
S(1^n)\models \varphi^{\inter I} $. Thus the r.h.s.\ in~\eqref{eq:mes} is
equivalent to $\str S(1^{\num(x)})\models \varphi^{\inter I}$ provided
$\num(x)\ge 2$, i.e., $x$ is non-empty. The l.h.s.\ in \eqref{eq:mes} is
equivalent to $1^{\num(x)}\in \un(Q)$. Thus $\varphi^{\inter I}$ witnesses
that $\un(Q)\in \AC^0$ according to Theorem~\ref{thm:BIS}. \proofend

\begin{rem}\label{rem:xac0}
The direction from left to right only required a $\AC^0$-circuit family for
instances of $\phalt_=$ with the fixed machine $\mathbb M^*$. This implies
that the assertions in Theorem~\ref{thm:phaltequal}~(i) are equivalent to
$\phalt_=\in \XAC{}$ (see Definition~\ref{df:xac}).
\end{rem}

\subsection{Almost tally problems}
Recall that a classical problem $Q\subseteq \{0,1\}^*$ is \emph{tally} if
$Q\subseteq\{1\}^*$. All parameterized problems mentioned in the introduction
are almost tally in the following sense:
\begin{defn}
A parameterized problem $(Q,\kappa)$ is \emph{almost tally} if
\[
Q\subseteq\big\{\angle{1^n,x} \mid n\in\N,x\in\{0,1\}^*\big\}
\]
and there is a computable $f:\N\to\N$ such that for all $n\in\N$,
$x\in\{0,1\}^*$
\[
|x|\le f(\kappa(\angle{1^n, x})).
\]
\end{defn}

Theorem~\ref{thm:phaltequal}~(ii) holds not only for $\phalt$ but for every
almost tally problem in $\para\NP$. In fact, $\phalt_=$ is the hardest almost
tally problem in $\para\NP$:
\begin{lem}\label{lem:phaltalmtally}
For every almost tally problem in $\para\NP$ there is an eventually definable
reduction to $\phalt_=$.
\end{lem}

\proof Let $(Q, \kappa)\in \para\NP$ be almost tally. The identity is a
parameterized reduction from $(Q, \kappa)$ to its re-parameterization $(Q,
\kappa')$ where $\kappa'(\angle{1^n, x}):=|x|$ for all $n\in \N$, $x\in
\{0,1\}^*$. We can therefore assume that $\kappa= \kappa'$.

Let $\mathbb M$ be a nondeterministic Turing machine that accepts $Q$ and on
input $\angle{1^n, x}$ runs in time at most $f(k)\cdot n^{c}$ where $c\in
\N$, $f: \N\to \N$ is a computable function, and $k:=|x|$.

Define $g:\N^2\to\N$ by
\[
g(m,k):= m^{c+1}+ 2m+ 2k+ 2.
\]
For $x\in\{0,1\}^*$ with $k:=|x|$, consider the nondeterministic Turing
machine $\mathbb M_x$\label{page:Mx} that on the empty input runs as follows:
\begin{center}
\fbox{
\begin{minipage}[t]{15cm}
\begin{algorithm}
\im0 nondeterministically write $\angle{1^m,x}$ for some $m\in \N$

\im0 simulate $\mathbb M$ on $\angle{1^m,x}$

\im0 \IF\ $\mathbb M$ does not halt or rejects, \THEN\ reject

\im0 make dummy steps such that so far the total
 running time  is $g(m,k)$

\im0 accept.
\end{algorithm}
\end{minipage}
}
\end{center}
Step~1 can be implemented to take exactly $2+2m+2+2k$ many steps
(recall~\eqref{eq:pairing}), so the dummy steps in line~4 are possible if
$m>f(k)$. Note that for each $k$, the function $m\mapsto g(m,k)$  is
injective. Thus, if $n>f(k)$, we have
\begin{eqnarray*}
\angle{1^n,x}\in Q
 & \Longleftrightarrow &
\angle{1^{g(n,k)+1},\mathbb M_x}\in\phalt_=.
\end{eqnarray*}
Choose a parameterized honest reduction from $(Q,\kappa)$ to $\phalt_=$ that
maps $\angle{1^n,x}$ with $n>f(k)$ to  $\angle{1^{g(n,k)+1}, \mathbb M_x}$.
We verify that it is eventually definable.

Choose a computable $h:\N\to\N$ and $d\in\N$ such that $n\ge h(k)$ implies
$n> f(k)$ and $n^d> g(n,k)+ 1, \num(x), \num(\mathbb M_x)$ for all $x\in
\{0,1\}^*$ of length $|x|= k$.

Let $\angle{1^n, x}$ satisfy $n> h(k)$ where $k:= |x|$. By
Lemma~\ref{lem:karyinterpretation}, $\str S(\angle{1^n, x})$ interprets the
expansion $(n^{d}, k, \num(x))$ of the standard $L^\rel_\arit$-structure
$n^d$ by two constants $c_0$ and $c_1$ denoting~$k$ and $\num(x)$. We are
left to show that $\big(n^{d}, k, \num(x)\big)$ interprets $\str
S(\angle{1^{g(n,k)+1}, \mathbb M_x})$. Using $\BIT$, it suffices to show that
both $g(n,k)+1$ and $\num(\mathbb M_x)$ are definable in
$\big(n^{d},k,\num(x)\big)$.

The former being clear, we show the latter. Consider a computable function
$M: \N\to \N$ that maps $\num(x)$ to $\num(\mathbb M_x)$. Choose a
quantifier-free $L^\rel_\arit$-formula $\varphi_M(x, y, \bar z)$ according to
Remark~\ref{rem:fxy}. We can assume that $h$ grows fast enough so that for
all $\ell\in\N$
\[
\N\models \exists \bar z {<} h(\ell)\ \varphi_M(\ell,M(\ell),\bar z).
\]
Then $\exists \bar z\ \varphi_M(c_1,y,\bar z)$ defines $\num(\mathbb M_x)$ in
the structure $(n^{d},k,\num(x))$. \proofend

It is straightforward to infer from Proposition~\ref{prop:ag} that
$\NE\subseteq\LINH$ if and only if every tally problem in $\NP$ is in
$\AC^0$. We don't know of a similarly easy  proof of the following
parameterized variant of this observation.  Instead, our proof relies on our
analysis of $\phalt_=$:

\begin{cor}\label{cor:almtallyNE}
$\NE\subseteq\LINH$ if and only if every almost tally problem in $\para\NP$
is in $\paraAC$.
\end{cor}

\begin{proof}
The l.h.s.\ is equivalent to $\phalt_=\in\paraAC$ by
Theorem~\ref{thm:phaltequal}~(i). And by Lemmas~\ref{lem:phaltalmtally} and
\ref{lem:ACclosed}, $\phalt_=\in\paraAC$ is equivalent to the r.h.s..
\proofend
\end{proof}

\subsection{Variants}
For the optimistic reader, Corollary~\ref{cor:almtallyNE} is an approach to
separate $\NE$ from~$\LINH$. From this perspective, it is of interest to ask
whether finding an almost tally problem outside $\paraAC$ but in a natural
subclass of $\para\NP$  implies stronger separations of natural complexity
classes. We verify the following variants of Corollary~\ref{cor:almtallyNE}:

\begin{lem}\label{lem:almtallyvariants} \
\begin{enumerate}\itemsep=0pt
\item[(i)] $\E\subseteq \LINH$ if and only if every almost tally problem in
    $\FPT$ is in $\paraAC$.

\item[(ii)] $\NLINSP\subseteq \LINH$ if and only if every almost tally
    problem in $\para\NL$ is in $\paraAC$.

\item[(iii)] $\LINSP\subseteq \LINH$ if and only if every almost tally
    problem in $\para\L$ is in $\paraAC$.
\end{enumerate}
\end{lem}

\proof The proof of (i) is analogous to the proof of
Corollary~\ref{cor:almtallyNE} using the subproblem of $\phalt_=$ where the
input machine $\mathbb M$ is deterministic. Similarly the proof of (iii) is
analogous to the proof of (ii). We show how (ii) is proved by modifying the
proof of Corollary~\ref{cor:almtallyNE}.

Consider the following variant of $\phalt$:
\npprob{$\phalt^*_=$}{$n,m\in \mathbb N$ in unary with $n\le m$ and a
nondeterministic Turing machine $\mathbb M$}{$|\mathbb M|$}{Does $\mathbb M$
accept the empty input in \emph{exactly} $n$ steps and space at most
$\floor{\log m}$?}
It is clear that this problem is in $\para\NL$.

\medskip
\noindent \textit{Claim 1.} $\phalt^*_=\in \paraAC$ if and only if
$\NLINSP\subseteq \LINH$.

\medskip
\noindent \textit{Proof of the Claim 1.} Assume $\NLINSP\subseteq \LINH$ and
let $Q$ be the classical problem underlying $\phalt^*_=$ but with the inputs
$n, m$ encoded in binary. Clearly, $Q\in \NLINSP \subseteq \LINH$, so
$\un(Q)\in \AC^0$ by Proposition~\ref{prop:ag}~(iii). Similarly as
Example~\ref{exa:internum} one sees that $\angle{1^n, 1^m, \mathbb M}\mapsto
1^{\num(\angle{n, m, \mathbb M})}$ is eventually definable. Then $\phalt^*_=
\in \paraAC$ follows as before.

Conversely, assume $\phalt^*_= \in \paraAC$ and let $Q\in \NLINSP$. Choose a
nondeterministic Turing machine $\mathbb M$ accepting $Q$ that on input $x\in
\{0,1\}^*$ runs in time at most
\[
\frac{\num(x)^{c}}{10c(|x|+2)} - 10c(|x|+ 2)- |x|
\]
and uses space at most $c\cdot |x|$; here $c\in \N$ is a suitable constant.
Define $\mathbb M^*$ as in page~\pageref{page:Mstar} but with the following implementation details. For
the simulation in line 2, first initialize a length $c(|y|+ 2)$ binary
counter using exactly $10c(|y|+ 2)$ steps, and increase it using exactly
$10c(|y|+ 2)$ many steps for each simulated step of $\mathbb M$. In line~4
continue increasing the counter in this way until it reaches
$\num(y)^c/(10c(|y|+ 2))$. For long enough $y$, the binary representation of
this number can be computed in time at most $\num(y)$ and space $O(|y|)$
(where the constant in the O-notation depends on $c$). This computation can
be done in parallel to the simulation in lines 2 and 4. Hence, line~5
completes exactly $\num(y)^c+1$ steps, and uses space at most $d\cdot |y|$
for a suitable~$d\ge c$.

Thus, we arrive at the following variant of~\eqref{eq:mes}. For long enough
$x\in\{0,1\}^*$:
 \begin{eqnarray*}
x \in Q  & \iff & \text{$\mathbb M^*$ accepts the empty input
in exactly $ \num(x)^c+1$ many steps}\\
&&\text{and space at most  $\floor{\log(\num(y)^d)}$}.
\end{eqnarray*}
Our assumption  $\phalt_=^*\in \paraAC$ implies that the classical problem
\begin{align*}
Q':=\big\{\angle{1^n,1^m}\bigmid\ & n\le m \text{ and $\mathbb M^*$ accepts the empty input
in exactly $n+1$ many steps}\\
&\text{and space at most $\floor{\log m}$}\big\}
\end{align*}
is in $\AC^0$. Now $\un(Q)\in\AC^0$ (and hence $Q\in\LINH$) follows as before
using an interpretation~$\inter I$ such that $\str S(1^n)^{\inter I}\cong
\str S(\big\langle 1^{n^c}, 1^{n^d}\big\rangle)$. \hfill$\dashv$

\medskip
\noindent \textit{Claim 2.} For every almost tally problem in $\para\NL$
there is an eventually definable reduction to $\phalt^*_=$.

\medskip
\noindent \textit{Proof of the Claim 2.} Let $(Q, \kappa)\in \para\NL$ be
almost tally and $\mathbb M$ be a nondeterministic Turing machine that
accepts $Q$ and that on input $\angle{1^n, x}$ runs in time at most
$f(k)\cdot n^{c}$ and space at most $f(k)+ c\cdot\log n$ where $c\in \N$,
$f:\N\to \N$ is a computable function, and $k:= \kappa(\angle{1^n, x})= |x|$.

For $x\in \{0, 1\}^*$ with $k:= |x|$, define the nondeterministic Turing
machine $\mathbb M_x$ as in page~\pageref{page:Mx} but
with a different $g$ (chosen below) and line 1 changed to
nondeterministically write some $m\in \N$ in binary in exactly
$2\ceil{\log(m+ 1)}+ 2$ steps. The simulation in line 2 is done as in the
previous claim maintaining a length $(c+1)\ceil{\log(m+1)}$ binary counter.
It further maintains the position of $\mathbb M$'s head on the input tape
\big(which we can assume to be at most $|\angle{1^m,x}|+1$\big) and uses it
to compute the currently scanned bit. If $k<m$, both the maintenance of the
counter and position can be done in exactly $10c\ceil{\log(m+1)}$ steps. In
line 4 the binary counter is updated until it reaches $m^{c+1}$. Hence line~4
is completed after exactly $g(m,k):=m^{c+1}\cdot 10c\ceil{\log(m+ 1)}+
2\ceil{\log(m+1)}+ 2$ steps. The dummy steps in line~4 are possible if
$m>f(k)$. In this case the computation takes space at most $d\log m$ for
suitable $d\in\N$. Thus, if $n> f(k)$, we have
\begin{eqnarray*}
\angle{1^n, x}\in Q
 & \Longleftrightarrow &
\angle{1^{g(n,k)+1},1^{n^d},\mathbb M_x} \in \phalt^*_=.
\end{eqnarray*}
Similarly as seen in the proof of Lemma~\ref{lem:phaltalmtally}, this implies
the claim. \hfill$\dashv$

\medskip
It now suffices to show that $\phalt^*_= \in \paraAC$ if and only if every
almost tally problem in $\para\NL$ is in $\paraAC$. The forward direction
follows from Claim~2 and Lemma~\ref{lem:ACclosed}. And if $\phalt^*_= \not\in
\paraAC$, then we get an almost tally problem in $\para\NL \setminus \paraAC$
by rewriting inputs $\angle{1^n, 1^m, \mathbb M}$ of $\phalt_=^*$ to
$\angle{1^{\angle{n, m}}, \mathbb M}$ where $\angle{n, m}$ is a pairing
function on $\N$. \proofend

We find it worthwhile to explicitly point out the following direct corollary
concerning the parameterized halting problem for {\em deterministic} Turing
machines:
\begin{cor}
If $\pdhalt\not\in \paraAC$, then $\E\not\subseteq \LINH$.
\npprob{$\pdhalt$}{$n\in \mathbb N$ in unary and a deterministic Turing
machine $\mathbb M$}{$|\mathbb M|$}{Does $\mathbb M$ accept the empty input
in at most $n$ steps?}
\end{cor}

\section{On the parameterized complexity of $\ptruth$}\label{sec:truth}

This section first observes that $\ptruth$ is ``the same'' as a basic
parameterized model-checking problem, uses this to prove the lower bound
$\ptruth\not\in\paraAC$ (Theorem~\ref{thm:truthAC}), and finally, based on
the previous section, infers consequences from {\em upper bounds} on the
parameterized complexity of $\ptruth$, including Theorem~\ref{thm:truthNP}.

\subsection{Model-checking arithmetic}
We observe that $\ptruth$ is ``the same'' as the parameterized model-checking
problem for first-order logic over finite standard $L^\rel_\arit$-structures:
\npprob{$\pMCFOA$}{$n\ge 2$ in unary and an $L^\rel_\arit$-sentence
$\varphi$}{$|\varphi|$}{$n\models \varphi$?}

\begin{lem}\label{lem:fcttorel}
There is a computable function that maps every $\Delta_0$-formula
$\varphi(x)$ to an $L^\rel_\arit$-sentence~$\psi$ such that for all $n\in \N$
with $n\ge 2$:
\begin{eqnarray}\label{eq:fcttorel}
\N\models \varphi(n)
 & \Longleftrightarrow &
 n\models \psi.
\end{eqnarray}

Further, there is a computable function that maps every
$L^\rel_\arit$-sentence $\psi$ to a $\Delta_0$-formula~$\varphi(x)$ such that
\eqref{eq:fcttorel} holds all $n\in\N$ with $n\ge2$.
\end{lem}

\proof
For the second
assertion define $\varphi(x)$ as~$\psi^{<x}$ with atoms rewritten in the functional language $L_\arit$. The first assertion
is folklore, see~\cite[Proposition~2.2]{gaifdim82}. We give a brief sketch
for completeness. It is routine to compute, given a $\Delta_0$-formula
$\varphi(\bar x)$, a constant $c_\varphi>1$ and an $L^\rel_\arit$-formula
$\psi_0(\bar x)$ such that
\begin{eqnarray*}
\N\models \varphi(\bar n)&\Longleftrightarrow&\N\models\psi_0^{<m}(\bar n)
\end{eqnarray*}
for all $\bar n,m\in\N$ with $m\ge \max\{\bar n,2\}^{c_\varphi}$. Hence, for
$n\ge 2$, the truth of $\varphi(n)$ is equivalent to $n^{c_\varphi} \models
\psi_0(n)$. Since the number $n$ is definable in  the standard
$L^\rel_\arit$-structure $n^{c_\varphi}$ (as the minimal element whose
$c_\varphi$-th power does not exist), we can replace $\psi_0(n)$ by some
sentence~$\psi_1$. Then set $\psi:=\psi_1^{\inter I_{c_\varphi}}$ for the
interpretation $\inter I_{c_\varphi}$ from
Lemma~\ref{lem:karyinterpretation}. \proofend

\subsection{A lower bound} In this subsection we prove Theorem~\ref{thm:truthAC}. We fix
a proper elementary extension $M$ of the standard $L_\arit^\rel$-model $\N$,
and a nonstandard element  $a\in M\setminus \N$. We let $<^M$ denote the
interpretation of $<$ in $M$. We need a simple lemma:
\begin{lem}\label{lem:Mf}
Let $f:\mathbb N\to \mathbb N$ be a computable function. Then there is an
$L_\arit^\rel$-formula $\chi_f(x, y)$ such that
for every $k\in \mathbb N$ and every $b<^Ma$:
\begin{eqnarray*}\label{eq:chif}
 f(k)=b & \iff & M \models \chi_f^{<a}(k,b).
 \end{eqnarray*}
\end{lem}

\proof
Let $\exists z\varphi(x,y,z)$ define $f$ in $\N$ where $\varphi$ is $\Delta_0$. As in the proof of Lemma~\ref{lem:fcttorel} let $\psi(x,y,z)$ be an $L_\arit^\rel$-formula such that for all $k,\ell,m,n\in\N$ with $k,\ell,m<n$:
\begin{eqnarray}\label{eq:fpsi}
\N\models\varphi(k,\ell,m)\iff\N\models\psi^{<n}(k,\ell,m)
\end{eqnarray}
Here, $\N$ ambiguously denotes the standard model in the respective languages. Set
$$
\chi_f(x,y):=\exists z\psi(x,y,z).
$$
If $f(k)=b$, then $b\in\N$ and there is $m\in\N$ such that $\N\models\psi^{<n}(k,b,m)$ for all $n>k,b,m$, so $M\models\psi^{<a}(k,b,m)$ as $k,b,m<^Ma$, and $M\models \chi^{<a}_f(k,b)$. Conversely, assume $M\models \chi^{<a}_f(k,b)$ for $b<^Ma$ and $f(k)\neq b$; then
$\exists u,y,z(k,y,z<u\wedge \psi^{<u}(k,y,z)\wedge y\neq f(k))$ holds in $M$ and hence in $\N$;  by \eqref{eq:fpsi},  $\N\models\varphi(k,\ell,m)$ for some $\ell,m\in\N$ with $\ell\neq f(k)$; this contradicts the fact that $\exists z\varphi(x,y,z)$ defines~$f$ in $\N$.
\proofend

Some notation: for $n\in\N$ define the $L^\rel_\arit$-formula $\text{``}x{=}n\text{''}$ by $\text{``}x{=}0\text{''}:= x{=}0$ and $\text{``}x{=}(n+1)\text{''}:=\exists y(\text{``}y{=}n\text{''}\wedge +(y,1,x))$.  For an $L^\rel_\arit$-formula $\varphi(y,\bar x)$ set $\varphi(\underline{n},\bar x):=\exists y(\text{``}y{=}n\text{''}\wedge\varphi(y,\bar x))$; we understand $\varphi^{<z}(\underline{n},\bar x)$ as
$\big(\varphi(\underline{n},\bar x)\big)^{<z}$. If $n<m$, then both
$(\text{``}x{=}n\text{''})^{<m}$ and $\text{``}x{=}n\text{''}$  define $n$ in $\N$, so $\varphi^{<m}(\underline{n},\bar x)$ and
$\varphi^{<m}(n,\bar x)$ are equivalent in $\N$. In particular, for
every $n\in\N$:
\begin{eqnarray}\label{eq:subst}
M\models\varphi^{<a}(\underline{n},\bar x)\leftrightarrow\varphi^{<a}(n,\bar x)  )
\end{eqnarray}

\medskip
\noindent{\em Proof of Theorem~\ref{thm:truthAC}:} For contradiction, assume
otherwise, so $\pMCFOA$ $\in \paraAC$ by Lemma~\ref{lem:fcttorel}.
By~Corollary~\ref{cor:evtfo}, there is an increasing computable function $h:
\mathbb N\to \mathbb N$ and a sentence $\textit{sat}$ such that for every
$n\in \mathbb N$ and every $L^\rel_\arit$-sentence $\varphi$ with $n>
h(\num(\varphi))$ we have
\begin{eqnarray}\label{eq:sat}
n \models \varphi
 & \iff &
\str S\big(\angle{1^n, \varphi}\big) \models \textit{sat}.
\end{eqnarray}

For $k<n$, let $(n,k)$ denote the expansion of the standard
$L^\rel_\arit$-structure $n$ that interprets a constant $c$ by $k$. It is
clear that there is an interpretation $\inter I$ (independent of $n,
\varphi$) such that $(n,\num(\varphi))^{\inter I}\cong\str S(\angle{1^n,
\varphi})$ for all $\varphi$ with $\num(\varphi)<n$. Replacing in
$\textit{sat}^{\inter I}$ the constant $c$ by a new variable $x$ gives an
$L^\rel_\arit$-formula $\textit{true}(x)$ such that for $n>
h(\num(\varphi))\ge \num(\varphi)$:
\begin{eqnarray}\nonumber
\str S\big(\angle{1^n, \varphi}\big) \models \textit{sat}
 & \iff &
n \models \textit{true}\big(\num(\varphi)\big) \\\label{eq:true}
 & \iff &\N\models \textit{true}^{<n}\big(\num(\varphi)\big),
\end{eqnarray}
where $\N$ is the standard $L^\rel_\arit$-model. Since $h:\mathbb N\to
\mathbb N$ is computable, there is an  $L^\rel_\arit$-formula $\text{``}
h(x)< y\text{''}$ with the obvious meaning. Note the l.h.s.\ of
\eqref{eq:sat} is equivalent to $\N\models\varphi^{<n}$. Combining
\eqref{eq:sat} and \eqref{eq:true} we get
\[
\N\models \text{``}h(\num(\varphi))< y\text{''}\to
\big(\varphi^{<y}
      \leftrightarrow \textit{true}^{<y}(\num(\varphi))\big)
\]
for every $L^\rel_\arit$-sentence $\varphi$. But $M\models \text{``}
h(\num(\varphi))<a\text{''}$, hence
\begin{eqnarray}\label{eq:Mvarphi}
M \models \varphi^{<a}
 \leftrightarrow \textit{true}^{<a}(\num(\varphi))
\end{eqnarray}
for every $L^\rel_\arit$-sentence $\varphi$. As stated in \cite[proof of
Proposition 3]{paris} this contradicts Tarski's undefinability of truth. We
include the details as they are omitted in~\cite{paris}.

The function which for every $L^\rel_\arit$-formula~$\varphi(x)$ maps
$\num(\varphi)$ to $\num(\varphi(\underline{\num(\varphi)}))$ is computable.
So by Lemma~\ref{lem:Mf}, there is a formula $\textit{sub}(x,y)$ such that
for every formula $\varphi(x)$ and every $b\in M$ with $b<^M a$:
\begin{eqnarray}\label{eq:sub}
b= \num(\varphi(\underline{\num(\varphi)}))
 & \iff &
M\models \textit{sub}^{<a}( \num(\varphi), b)
\end{eqnarray}
Define $\chi(x) := \forall y\big(\textit{sub}(x,y) \to \neg
\textit{true}(y))$ and $\theta:=\chi(\underline{\num(\chi)})$, and note
\begin{equation}\label{eq:numtheta}
\num(\theta)= \num(\chi(\underline{\num(\chi)})).
\end{equation}
We arrive at the desired contradiction:
\[
\begin{array}{lrll}
 & M \models \theta^{<a} \iff
  & M\models \forall y{<}a \big(\textit{sub}^{<a}(\num(\chi),y)\to \neg \textit{true}^{<a}(y)\big)
   & \text{by~\eqref{eq:subst}} \\
 & \iff &
  M \models \textit{sub}^{<a}(\num(\chi),b) \to \neg \textit{true}^{<a}(b)
   & \text{for all $b<^M a$} \\
 & \iff &
  M\models \neg \textit{true}^{<a}(\num(\theta))
   & \text{by \eqref{eq:sub} and \eqref{eq:numtheta}} \\
 & \iff & M\not\models \theta^{<a}
  & \text{by~\eqref{eq:Mvarphi}}.
\end{array}
\]

\vspace{-.7cm} \proofend

\subsection{Upper bounds}
Based on our analysis of halting problems in Section~\ref{sec:halt}, we now
see that various \emph{upper bounds} on the complexity of $\ptruth$ imply
separations of classical complexity classes from $\LINH$. This is our main
result. The first assertion is Theorem~\ref{thm:truthNP}:
\begin{theo}\label{thm:truthupper}\
\begin{enumerate}\itemsep=0pt
\item[(i)] If $\ptruth\in \para\NP$, then $\NE\not\subseteq \LINH$.

\item[(ii)] If $\ptruth\in \FPT$, then $\E\not\subseteq \LINH$.

\item[(iii)] If $\ptruth\in \para\NL$, then $\NLINSP\not\subseteq \LINH$.

\item[(iv)] If $\ptruth\in \para\L$, then $\LINSP\not\subseteq \LINH$.
\end{enumerate}
\end{theo}

\proof Since $\ptruth$ is an almost tally problem, (i) follows from
Theorem~\ref{thm:truthAC} and Corollary~\ref{cor:almtallyNE}. The other
assertions follow using Lemma~\ref{lem:almtallyvariants}. \proofend


\section{Provability of the \MRDP\ theorem}\label{sec:mrdp}

In this section we prove Theorem~\ref{thm:mrdp} as a corollary to
Theorem~\ref{thm:truthNP} via Parikh's theorem~\cite{par71}:
\begin{theo}\label{thm:parikh}
Let $T$ be a $\Pi_1$-theory and $\varphi(\bar x, \bar y)$ a
$\Delta_0$-formula. If $T$ proves $\exists \bar y\; \varphi(\bar x, \bar y)$,
then $T$ proves $\exists \bar y{<}p(\bar x)\; \varphi(\bar x, \bar y) $ for
some term $p(\bar x)$.
\end{theo}

Here, a \emph{theory} is a set of sentences, and a  {\em $\Pi_1$-theory} is a
set of $\Pi_1$-sentences. For example, $I\Delta_0$ is (equivalent to) a
$\Pi_1$-theory.

\begin{defn}
A theory \emph{$T$ proves \MRDP} if for every $\Delta_0$-formula
$\varphi(\bar x)$ there are $L_\arit$-terms $p(\bar x, \bar y)$ and $q(\bar
x, \bar y)$ such that $T$ proves
\begin{align*}
\varphi(\bar x)
 \leftrightarrow \exists \bar y\; p(\bar x, \bar y){=} q(\bar x, \bar y).
\end{align*}
\end{defn}

As mentioned in the introduction it is a long standing open problem whether
$I\Delta_0$ proves \MRDP\ but it is known that adding exponentiation
suffices. Intuitively, the following concept asks whether \MRDP\ can be
proved using exponentiation only once.

\begin{defn}\label{def:smallnumber}
A theory \emph{$T$ proves \MRDP\ for small numbers} if for every $k\in
\mathbb N$ and every $\Delta_0$-formula $\varphi(\bar x)= \varphi(x_0,
\ldots, x_{k-1})$ there are $L_\arit$-terms $p(\bar x, \bar y)$ and $q(\bar
x, \bar y)$ such that $T$ proves
\begin{eqnarray}\label{eq:smallmrdp}
\textstyle\bigwedge_{i<k}  2^{x_i}{\le}z
 & \to & \Big(\varphi(\bar x) \leftrightarrow
  \exists \bar y\; p(\bar x, \bar y){=} q(\bar x, \bar y)\Big).
\end{eqnarray}
\end{defn}

\noindent Here, $2^x{\le}y$ stands for a well-known
$\Delta_0$-formula~\cite[Section~V.3.(c)]{hp}. The following strengthens
Theorem~\ref{thm:mrdp}:
\begin{theo}
Let $T$ be a true $\Pi_1$-theory. Moreover, assume that $T$ is computably
enumerable. If $T$ proves \MRDP\ for small numbers, then
 $\ptruth\in\para\NP$
  and thus $\NE\not\subseteq \LINH$.
\end{theo}

\proof
Assume $T$ proves \eqref{eq:smallmrdp} for $\varphi(x)$, and hence
\[
2^{x}{\le}z
 \wedge \varphi(x) \to
  \exists \bar y\; p(x, \bar y){=} q(x, \bar y).
\]
By Theorem~\ref{thm:parikh} $\exists\bar y$ can be replaced by $\exists \bar
y{<}r(x,z)$ for some term $r(x,z)$. But since $T$ proves \eqref{eq:smallmrdp}
for $\varphi(x)$, $T$ proves
\[
2^{x}{\le}z
 \to\big( \varphi(x) \leftrightarrow
  \exists \bar y{<}r(x,z)\; p(x, \bar y){=} q(x, \bar y)\big).
\]

Since $T$ is computably enumerable, such terms $p,q,r$ can be computed from
$\varphi$. Given an instance $\angle{1^n,\varphi}$ of $\ptruth$, compute
$p,q,r$ as above, guess $\bar m<r(n,2^n)$ and check $p(n, \bar m){=} q(n,
\bar m)$. Note the length of the guess $\bar m$ is $O(|r|\cdot|\bar y|\cdot
n)$. The check can be done in time $(|p|\cdot |q|\cdot |r|\cdot n)^{O(1)}$.
Thus, $\ptruth\in\para\NP$. Now apply Theorem~\ref{thm:truthNP}. \proofend

It would be interesting to find variants of the this result that infer
$\ptruth\in \FPT$ or $\ptruth\in \para\NL$ from certain provabilities of
\MRDP\ or other arithmetical statements. Note this implies stronger
separations of complexity classes by Theorem~\ref{thm:truthupper}.



\section{\phalt\ and a universal $\AC^0$-easy set in $\NP$}\label{sec:immu}

In this section we prove Theorem~\ref{thm:immu}. We use the following
technical lemma stating, roughly, that every computable function is dominated
by a computable injection which is $\AC^0$-invertible.

\begin{lem}\label{lem:h}
Let $f: \mathbb N\to \mathbb N$ be computable. Then there is an
increasing $h:\mathbb N\to \mathbb N$ with the following properties.
\begin{enumerate}\itemsep=0pt
\item[(i)] $h(n)\ge f(n^2)$ for every $n\in \mathbb N$.

\item[(ii)] $1^n\mapsto 1^{h(n)}$ is computable in time
    $h(n)^{O(1)}$.

\item[(iii)] There is an $L^\rel_\arit$-sentence $\varphi_h$ such that for
    every $x\in\{0,1\}^*$ with $|x|\ge 2$:
    \begin{eqnarray*}
    \str S(x)\models\varphi_h
     & \iff &
      x= 1^{h(n)} \text{for some $n\in\N$}.
    \end{eqnarray*}

\item[(iv)] There is an $L^\rel_\arit$-formula $\varphi(x)$ that defines
    $n$ in $ \str S(1^{h(n)})$ for every $n\ge 2$.
\end{enumerate}
\end{lem}

\proof Given a deterministic Turing machine $\mathbb M$ and $x\in \{0,1\}^*$
we let $y_{\mathbb M, x}\in\{0,1\}^*$ encode the computation of $\mathbb M$
on $x$. This encoding can be chosen so that:
\begin{enumerate}\itemsep=0pt
\item[(a)] $x\mapsto y_{\mathbb M, x}$ is computable in time $|y_{\mathbb
    M, x}|^{O(1)}$.

\item[(b)] $\big\{\angle{x, y_{\mathbb M, x}}\bigmid x\in
    \{0,1\}^*\big\}\in \AC^0$.
\end{enumerate}
Now let $\mathbb M_f$ be a Turing machine that computes $1^n\mapsto
1^{f(n)}$. Let $\mathbb M$ be the machine that on input~$1^n$ runs $\mathbb
M_f$ on $1^{i^2}$ for every $i\le n$. Define the \emph{increasing} function
$h:\mathbb N\to \mathbb N$ by
\begin{equation}\label{eq:hdefinition}
h(n)= \num\big(\angle{1^n, y_{\mathbb M, 1^n}}\big)
\end{equation}
Clearly, the string $y_{\mathbb M_f, 1^{n^2}}$ encoding the computation of
$\mathbb M_f$ on input $1^{n^2}$ has length at least $f(n^2)$. Similarly,
$|y_{\mathbb M, 1^n}|\ge f(n^2)$. Thus $h(n)\ge f(n^2)$ for every $n\in
\mathbb N$, i.e., (i) holds.

(ii) holds by (a). To show (iii), Theorem~\ref{thm:BIS} and (b) imply that
there is an $L^\rel_\arit$-sentence $\varphi$ that holds precisely in the
string structures of the form $\str S\big(\bin(h(n))\big)$ for $n\in\N$.
Using $\BIT$, there is an interpretation $\inter I$ such that $\str
S(1^m)^{\inter I}\cong\str S(\bin(m))$ for every $m> 2$, so
$\varphi_h:=\varphi^{\inter I}$ holds precisely in the string structures of
the form $\str S(1^{h(n)})$ for $n\in\N$ (we have $h(n)>2$ for all $n\in\N$).

Trivially, $n$ is definable in $\str S(\bin(h(n)))$, so (iv) follows using
the interpretation $I$ above. \proofend

Theorem~\ref{thm:immu} is an easy consequence of the following stronger
result, and we view it as good evidence for the truth of
Conjecture~\ref{conj:phaltAC}.

\begin{theo}\label{thm:universal}
Assume $\phalt\in \paraAC$. Then there is an infinite tally problem $X$ such
that for every $Q\in\NP$ we have $Q\cap X\in \AC^0$.
\end{theo}

\noindent \textit{Proof of Theorem~\ref{thm:immu} from
Theorem~\ref{thm:universal}:} Assume $\phalt\in \paraAC$ and let $Q\in  \NP$.
Let $X$ be as stated in Theorem~\ref{thm:universal}. Then either $Q\cap X$ or
$(\{0,1\}^*\setminus Q)\cap X$ is infinite. By Theorem~\ref{thm:universal}
they are both in $\AC^0$. Hence $Q$ is not $\AC^0$-bi-immune. \proofend

\noindent \textit{Proof of Theorem~\ref{thm:universal}:}
By~Corollary~\ref{cor:evtfo} there is a computable increasing
function $f:\mathbb N\to \mathbb N$ and an
$L^\rel_\arit$-sentence $\varphi$ such that for every $\angle{1^n, \mathbb
M}$ with $n \ge f(|\mathbb M|)$:
\begin{eqnarray} \label{eq:phaltvarphi}
\str S\big(\angle{1^n, \mathbb M}\big)\models \varphi
&  \iff &\text{$\mathbb M$ accepts the empty input tape in at most $n$ steps}.
\end{eqnarray}
Now let $h:\mathbb N\to \mathbb N$ be the increasing function as stated in
Lemma~\ref{lem:h}. In particular, there is a deterministic Turing machine
$\mathbb M_h$ and a constant $c\ge 1$ such that on input $1^m$ the machine
$\mathbb M_h$ outputs the string $1^{h(m)}$ in time $h(m)^c$. The desired
set $I$ is defined by
\[
X:= \big\{1^{h(m)} \mid m\ge 2\big\}.
\]
By Lemma~\ref{lem:h}~(iii) the sentence $\varphi_h$ witnesses $X\in\AC^0$
according to Theorem~\ref{thm:BIS}.

Now let $Q\subseteq \{0,1\}^*$ be a problem in $\NP$. In particular, there is
a nondeterministic Turing machine $\mathbb M_Q$ and a constant $d\ge 1$ such
that $\mathbb M_Q$ accepts $x$ in time $|x|^d$.

Define the nondeterministic Turing machine $\mathbb M_{Q,h,m}$ to run
$\mathbb M_h$ on $1^m$ to produce output $1^{h(m)}$ and then run $\mathbb
M_Q$ on $1^{h(m)}$. This machine runs in time
\[
n(m):= h(m)^c + h(m)^{d}.
\]
Choose a constant $e\in\N$ such that $m\ge |\mathbb M_h|+ |\mathbb
M_Q|+e$ implies $m^2\ge |\mathbb M_{Q,h,m}|$. Then
\begin{align*}
n(m) \ge h(m)\ge f(m^2)
  \ge f(|\mathbb M_{Q,h,m}|).
\end{align*}
Hence, by \eqref{eq:phaltvarphi}, for $m\ge |\mathbb M_h|+ |\mathbb
M_Q|+e$:
\begin{eqnarray}\notag
 1^{h(m)}\in Q & \iff &
 \text{$\mathbb M_{Q,h,m}$ accepts the empty input in at most $n(m)$ steps} \\\label{eq:Mm}
 & \iff &
\str S(\angle{1^{n(m)},\mathbb M_{Q,h,m}})\models\varphi.
\end{eqnarray}

Lemma~\ref{lem:h}~(iv) implies that there is an interpretation $\inter I$
such that for every $m\in \mathbb N$
\[
\str S(1^{h(m)})^{\inter I}
 = \str S\big(\angle{1^{n(m)}, \mathbb M_{Q,h,m}}\big).
\]

By Theorem~\ref{thm:BIS} it suffices to show that for every $x\in \{0,1\}^*$
with $|x|\ge h(|\mathbb M_h|+ |\mathbb M_Q|+ e)$:
\begin{eqnarray*}
x\in Q\cap X
 & \iff &
\str S(x)\models \varphi_h\wedge \varphi^{\inter I}.
\end{eqnarray*}

Assume $x\in Q\cap X$. Then $x=1^{h(m)}$ for some $m\ge 2$ and $\str
S(x)\models\varphi_h$. Since $|x|=h(m)\ge h(|\mathbb M_h|+ |\mathbb M_Q|+e)$
and $h$ is increasing, we have $m\ge |\mathbb M_h|+ |\mathbb M_Q|+e$. Thus
$x=1^{h(m)}\in Q$ implies $\str S(\angle{1^{n(m)}, \mathbb M_{Q,h,m}})\models
\varphi$ by~\eqref{eq:Mm}, and $\str S(1^{h(m)})\models \varphi^{\inter I}$
follows.

Conversely, assume $S(x)\models\varphi_h\wedge\varphi^{\inter I}$. By  $\str
S(x)\models\varphi_h$, we have $x\in X$, so $x=1^{h(m)}$ for some $m\ge 2$.
By $\str S(1^{h(m)})\models\varphi^{\inter I}$ we have $\str
S(\angle{1^{n(m)}, \mathbb M_{Q,h,m}})\models\varphi$. This implies
$x=1^{h(m)}\in Q$ by~\eqref{eq:Mm} because, as above, $m\ge |\mathbb M_h|+
|\mathbb M_Q|+e$. \proofend

\section{Problem comparison}\label{sec:concl}
\subsection{The role of uniformity}
Our proof of the lower bound $\ptruth\notin \paraAC$
(Theorem~\ref{thm:truthAC}) makes crucial use of the uniformity condition in
the definition of $\paraAC$. To shed some light on this dependence, we relax
the uniformity condition as follows.

\begin{defn}\label{df:xac} Let $(Q,\kappa)$ be a parameterized problem  and $d,k\in\N$. The {\em $k$-th slice of $(Q,\kappa)$} is the classical problem $\{x\in Q\mid\kappa(x)=k\}$.
The class $\XAC{}$ contains $(Q,\kappa)$ if and only if $\AC^0$ contains every slice of $(Q,\kappa)$.
The class $\XAC{d}$ contains $(Q,\kappa)$ if and only if $\CAC d$ contains every slice of $(Q,\kappa)$; here,
 $\CAC d$ denotes the class of problems decided by
dlogtime uniform circuit families of polynomial size and depth $d$.
\end{defn}

Clearly,
\begin{equation}\label{eq:xac}
\textstyle
\paraAC\subseteq \bigcup_{d\in\N}\XAC d\subseteq\XAC{}
\end{equation}
and $\XAC 1\not\subseteq \paraAC$ as witnessed, e.g., by an undecidable
problem with parameterization~$\num$.
The class $\XAC{}$ is important in our context because it is a natural upper bound on $\ptruth$:

\begin{prop}\label{prop:truthxac} $\ptruth\in\XAC{}$.
\end{prop}
\proof It suffices to show that for every $\Delta_0$-formula $\varphi(x)$
the problem $\{1^n\mid\N\models\varphi(n)\}$
belongs to~$\AC^0$. But this problem is $\un(Q)$ for $Q:=\{x\in\{0,1\}^*\mid \N\models\varphi(\num(x))\}$. Clearly  $Q\in\LINH$, so
$\un(Q)\in\AC^0$ follows from Proposition~\ref{prop:ag}.
\proofend

We show that it is likely difficult to improve
Theorem~\ref{thm:truthAC}
 to $\ptruth
\not\in\bigcup_{d\in\N}\XAC d$. This somewhat artificial class also exhibits
the different behaviors of the parameterized problems $\phalt$,
$\phalt_=$, and $\ptruth$.

\begin{theo}\label{thm:xac}\
\begin{enumerate}\itemsep=0pt
\item[(i)] $\phalt\in \XAC 2$.

\item[(ii)] $\phalt_=\in \XAC d$ for some $d\in\N$ if and only if
    $\NE\subseteq \LINH$.

\item[(iii)] $\ptruth\in \XAC d$ for some $d\in\N$ if and only if $\LINH$
    collapses.
\end{enumerate}
\end{theo}

\proof  (i) has been observed in the introduction and (ii) follows from
Remark~\ref{rem:xac0} and \eqref{eq:xac}. To see (iii), assume $\LINH$ collapses. Paris and
Dimitracopolous~\cite[Proof of Proposition 4]{paris} showed that this implies
the following. There is an $L^\rel_\arit$-formula $\lambda(x,y)$ such that
for every $\Delta_0$-formula $\varphi(x)$ there are
$c_\varphi,d_\varphi,e_\varphi\in \N$ such that for all $n\ge c_\varphi$
\begin{eqnarray*}
\N\models \varphi(n)
 & \iff & n^{d_\varphi}\models\lambda(n,e_\varphi)
\end{eqnarray*}
For each fixed $\varphi$ there is an $\AC^0$-family that given $1^n$ decides
whether $n$ satisfies the the r.h.s.. The size of this family is bounded
$n^{f_\varphi}$ for some $f_\varphi\in\N$ depending on $\varphi$, but the
depth of this family  is determined by the quantifier alternation rank of
$\lambda$ and, in particular, does not depend on $\varphi$. This implies
$\ptruth\in \XAC d$ for some $d\in\N$.

Conversely, assume $\ptruth\in \XAC d$ and let $Q\in\LINH$. It is well
known (see e.g.\ \cite[Ch.V, Lemma 2.13]{hp}) that
there is a $\Delta_0$-formula that is satisfied by $\num(x)$ if and only if
$x\in Q$. Fixing this formula in the input to $\ptruth$, the assumption
implies that there is a dlogtime uniform circuit family $(C_{n})_n$ of
polynomial size and depth $d$ such that for all  $x\in\{0,1\}^*$:
\begin{eqnarray*}
x\in Q & \iff & C_{\num(x)}(1^{\num(x)})=1.
\end{eqnarray*}

It suffices to show that, given $x$, the r.h.s. can be checked by an
alternating machine in linear time with $d$ alternations. This is
straightforward by guessing a path through $C_{\num(x)}$. E.g., if the output
gate is a $\vee$-gate, the machine existentially guesses an input gate $g_1$
to it, and if it is a $\wedge$-gate it universally guesses $g_1$. Depending
on the type of $g_1$ it either existentially or universally guesses an input
gate $g_2$ to $g_1$, and so on. When reaching (with $g_{d-1}$ or earlier) an
input gate or a negation thereof, the machine checks  it is satisfied by the
corresponding bit of $x$. Each guess requires~$O(|x|)$ bits. Checking that
e.g.~$g_2$ is an input to $g_1$ can be done in time logarithmic in the size
of~$C_{\num(x)}$, that is, in time~$O(|x|)$. We omit further
details.\proofend

\subsection{Reducibilities}
In this section we draw some corollaries  concerning how our problems
$\phalt,\phalt_=$ and $\ptruth$ compare with respect to our notion of
reducibility. Saying that a (parameterized) problem is {\em reducible} to
another means that there is an eventually definable reduction. Two problems
are {\em equivalent} if they are reducible to one another.

The picture is as follows: an arrow indicates reducibility, $\equiv$ means
equivalence.
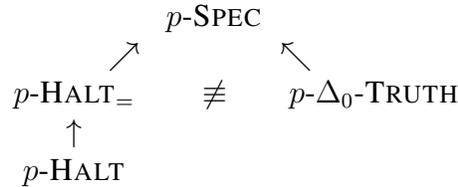
\[
\begin{array}{ccc}
&\pspec&\\
\hspace*{7ex}\nearrow&&\hspace{-11ex}\nwarrow\\
\phalt_=&\not\equiv&\ptruth\\
\uparrow&&\\
\phalt&&
\end{array}\]
In particular, we show unconditionally that $\phalt_=$ and $\ptruth$ are  not
equivalent and both are reducible to yet another almost tally problem of
central importance to mathematical logic, namely the following parameterized
version of the spectrum problem:
\npprob{$\pspec$}{$n\in \mathbb N$ in unary and a first-order sentence $\varphi$}
{$|\varphi|$}{Does $\varphi$ have a model of size $n$?}
Recall that having a model of size $n$ means that $n$ belongs to the spectrum
of $\varphi$.

\medskip
We start comparing $\phalt$ and $\phalt_=$. Clearly, $\phalt$ is reducible to
$\phalt_=$. Concerning the converse we note that Theorem~\ref{thm:xac}~(i),
(ii) implies:

\begin{cor}
If $\phalt_=$ is reducible to $\phalt$, then $\NE\subseteq\LINH$.
\end{cor}

Adapting a mode of speech from~\cite{cfslice}, call an almost tally problem
$(Q,\kappa)$ {\em slicewise monotone} if $(1^n,x)\in Q$ implies $(1^m,x)\in
Q$ for all $x\in\{0,1\}^*$ and all $n,m\in\N$ with $n<m$.
One can show that $\phalt$ is the hardest such problem in $\para\NP$. This is
an easy modification of the proof Lemma~\ref{lem:phaltalmtally} and
strengthens \cite[Proposition~11]{cfslice}:

\begin{cor} Every almost tally problem in $\para\NP$ that is slicewise monotone is reducible to~$\phalt$.
\end{cor}

We turn to  $\phalt_=$ and $\ptruth$.

\begin{cor} \
\begin{enumerate}\itemsep=0pt
\item[(i)] If $\ptruth$ is reducible to $\phalt_=$, then $\NE\not\subseteq\LINH$.
\item[(ii)]  If $\phalt_=$ is reducible to $\ptruth$, then $\NE\subseteq\LINH$.
\item[(iii)] $\ptruth$ and $\phalt_=$ are not equivalent.
\end{enumerate}
\end{cor}

\proof  (iii) follows from (i) and (ii).
For (i), assume $\ptruth$ is reducible to $\phalt_=$. Then $\ptruth\in\para\NP$ and
$\NE\not\subseteq\LINH$ follows by Theorem~\ref{thm:truthNP}.

For (ii), assume $\phalt_=$ is reducible to $\ptruth$. Then $\phalt_=\in\XAC{}$  by Proposition~\ref{prop:truthxac} and hence $\NE\subseteq\LINH$ by Remark~\ref{rem:xac0}.
\proofend

Finally, we turn to $\pspec$:

\begin{prop} Both $\phalt$ and $\ptruth$ are reducible to $\pspec$.
\end{prop}

\proof It is straightforward to compute from a nondeterministic Turing machine $\mathbb M$ a first-order sentence $\varphi_{\mathbb M}$ that has a model of size $n$ if and only if $\mathbb M$ accepts the empty input in exactly $n$ steps.

Concerning $\ptruth$, by Lemma~\ref{lem:fcttorel}, it suffices to show that $\pMCFOA$ is
reducible to $\pspec$: map an instance $(1^n,\varphi)$ of $\pMCFOA$ to $(1^n,\varphi\wedge\psi)$ where $\psi$ is
an $L^\rel_\arit$-sentence whose finite models are exactly those isomorphic to some standard finite $L^\rel_\arit$-structure.
\proofend

Observe $\pspec$ can be solved in nondeterministic time $n^{f(k)}$ for some computable $f:\N\to\N$ where $k:=|\varphi|$ is the parameter. Can the parameter be moved out of the exponent? We find in worthwhile to explicitly point out the following direct corollary of the previous proposition and Theorem~\ref{thm:truthNP}:

\begin{cor} If $\pspec\in\para\NP$, then $\NE\not\subseteq\LINH$.
\end{cor}

\bibliographystyle{plain}
\bibliography{MRDP}

\end{document}